\documentclass[a4paper]{article}

\usepackage[utf8x]{inputenc}

\usepackage{amsmath}

\usepackage{amsthm}
\usepackage{amssymb}
\usepackage{breqn}
\usepackage{array}
\usepackage{enumerate}
\usepackage{bookmark}
\usepackage{hyperref}
\hypersetup{
  pdftitle = {Exact quantum query complexity of EXACT and THRESHOLD},
  pdfauthor = {},
  pdfkeywords = {},
  pdfstartview=FitH,
  unicode=true
}

\newcommand{\ket}[1]{\left| #1 \right\rangle}
\newcommand{\bra}[1]{\left\langle #1 \right|}

\newtheorem{definition}{Definition}
\newtheorem{theorem}{Theorem}
\newtheorem{corollary}{Corollary}
\newtheorem{proposition}{Proposition}

\begin{document}

\title{Exact quantum query complexity of EXACT and THRESHOLD}
\author{\texorpdfstring{Andris Ambainis \and Jānis Iraids \and Juris Smotrovs}{Andris Ambainis, Jānis Iraids, Juris Smotrovs}\\
University of Latvia, Raiņa bulvāris 19, Riga, LV-1586, Latvia}



\maketitle

\bookmarksetup{startatroot}

\begin{abstract}
A quantum algorithm is {\em exact} if it always produces the correct
answer, on any input. Coming up with exact quantum algorithms that 
substantially outperform the best classical algorithm has been 
a quite challenging task.

In this paper, we present two new exact quantum algorithms 
for natural problems: 
\begin{itemize}
\item 
for the problem EXACT$_k^n$ in which we have to determine whether
the sequence of input bits $x_1, \ldots, x_n$ contains exactly $k$ values
$x_i=1$; 
\item
for the problem THRESHOLD$_k^n$ in which we have to determine
if at least $k$ of $n$ input bits are equal to 1.
\end{itemize}
\end{abstract}

\section{Introduction}

We consider quantum algorithms in the query model.
The algorithm needs to compute a given Boolean function
$f:\,\{0,1\}^n\to\{0,1\}$ by querying its input bits until it is able
to produce the value of the function, either with certainty, or
with some error probability. The complexity of the algorithm
is measured as the number of queries it makes (other kinds of computation
needed to produce the answer are disregarded).

In the {\em bounded error} setting where the algorithm is allowed to give an 
incorrect answer with probability not exceeding a given constant $\epsilon$, $0<\epsilon<\frac12$, many efficient quantum algorithms are known, with
either a polynomial speed-up over classical algorithms (e.g.,
\cite{Gro96,Amb07,FGG08,RS08,Bel12}), or, in the case
of partial functions, even an exponential speed-up (e.g., 
\cite{Sim97,Shor97}).

Less studied is the {\em exact} setting where the algorithm
must give the correct answer with certainty. Though for partial functions
quantum algorithms with exponential speed-up are known (for instance,
\cite{DJ92,BH97}), the results for total functions up to recently have
been much less spectacular: the best known quantum speed-up was just by a factor
of 2.

Even more, as remarked in \cite{MJM11},
all the known algorithms achieved this speed-up by the same trick: exploiting
the fact that XOR of two bits can be computed quantumly with one query, while
a classical algorithm needs two queries \cite{DJ92,CEMM98,FGGS98}.

A step forward was made by \cite{MJM11} which presented
a new algorithm achieving the speed-up by a factor of 2,
without using the ``XOR trick''. The algorithm is for the Boolean function
EXACT$_2^4$ which is true iff exactly 2 of its 4 input bits are equal to 1.
It computes this function with 2 queries, while a classical (deterministic)
algorithm needs 4 queries.

This function can be generalized to EXACT$_k^n$ in the obvious way.
Its deterministic complexity is $n$ (due to its sensitivity being $n$, 
see \cite{NS94}). \cite{MJM11} conjectured that its quantum query complexity
is $\max{\{k,n-k\}}$.

In this paper we prove the conjecture.
We also solve the problem for a similar function, THRESHOLD$_k^n$ which
is true iff {\em at least} $k$ of the input bits are equal to 1.
When $n=2k-1$, this function is well-known as the MAJORITY function.
The quantum query complexity of THRESHOLD$_k^n$ turns out to be $\max{\{k,n-k+1\}}$,
as conjectured in \cite{MJM11}.

In a recent work \cite{Amb12}, a function $f(x_1, \ldots, x_n)$ with
the deterministic query complexity $n$ and the exact quantum query complexity
$O(n^{.8675...})$ was constructed. The quantum advantage that is achieved by our 
algorithms is smaller but we think that our results are still interesting, for 
several reasons. 

First, we present quantum algorithms for computational
problems that are natural and simple to describe. Second, our algorithms contain new ideas 
which may be useful for designing other exact algorithms. Currently, the toolbox of
ideas for designing exact quantum algorithms is still quite small. Expanding it is 
an interesting research topic.

\section{Technical Preliminaries}

We denote $[m]=\{1,2,\ldots,m\}$.
We assume familiarity with basics of quantum computation \cite{NC00}.
We now briefly describe the quantum query algorithm model.

Let $f:\,\{0,1\}^n\to\{0,1\}$ be the Boolean function to compute,
with the input bit string $x=x_1x_2\ldots x_n$. The quantum query algorithm
works in a Hilbert space with some fixed basis states. It starts in a
fixed starting state, then performs on it a sequence of unitary transformations
$U_1$, $Q$, $U_2$, $Q$, \dots, $U_t$, $Q$, $U_{t+1}$.
The unitary transformations $U_i$ do not depend on
the input bits, while $Q$, called the {\em query transformation}, does,
in the following way. Each of the basis states corresponds to either one or none
of the input bits. If the basis state $\ket\psi$ corresponds to the $i$-th
input bit, then $Q\ket\psi=(-1)^{x_i}\ket\psi$. If it does not correspond to any
input bit, then $Q$ leaves it unchanged: $Q\ket\psi=\ket\psi$. For convenience
in computations, we denote $\hat x_i=(-1)^{x_i}$.

Finally, the algorithm performs a full measurement in the standard basis. 
Depending on the result of the measurement, it outputs either 0 or 1
which must be equal to $f(x)$.

By the principle of delayed measurement, sometimes a measurement performed 
in the middle of computation is equivalent to it being performed at the end
of computation \cite{NC00}. We will use that in our algorithms, because
they are most easily described as recursive algorithms with the following
structure: perform unitary $U_1$, query $Q$, unitary $U_2$, then measure;
depending on the result of measurement, call a smaller (by 2 input bits) 
instance of the algorithm. The principle of delayed measurement ensures that
such recursive algorithm can be transformed by routine techniques into the
commonly used query algorithm model described above.

The minimum number of queries made by any quantum algorithm computing $f$
is denoted by $Q_E(f)$. We use $D(f)$ to denote the minimum number of
queries used by a deterministic algorithm that computes $f$.

\section{Algorithm for EXACT}

\begin{definition}
The function $EXACT_k^n$ is a Boolean function of $n$ variables being true iff \emph{exactly} $k$ of the variables are equal to $1$.
\end{definition}

\begin{theorem}
\[Q_E(EXACT_k^{2k}) \leq k\]
\end{theorem}
\begin{proof}
We present a recursive algorithm. When $k=0$ the algorithm returns $1$ without making any queries. Suppose $k=m$. For the recursive step we will use basis states
$\ket 0$, $\ket 1$, \dots, $\ket n$ and $\ket{i,j}$ with $i,j\in[2m]$, $i<j$.
The $i$-th input bit will be queried from the state $\ket i$.
We begin in the state $\ket{0}$ and perform a unitary transformation $U_1$:
\[U_1\ket{0} \rightarrow \sum_{i=1}^{2m}{\frac{1}{\sqrt{2m}}\ket{i}}.\]
Next we perform a query:
\[\sum_{i=1}^{2m}{\frac{1}{\sqrt{2m}}\ket{i}} \xrightarrow{Q} \sum_{i=1}^{2m}{\frac{\hat x_i}{\sqrt{2m}}\ket{i}}.\]
Finally, we perform a unitary transformation $U_2$, such that
\[U_2\ket{i} = \sum_{j>i}{\frac{1}{\sqrt{2m}}\ket{i,j}} - \sum_{j<i}{\frac{1}{\sqrt{2m}}\ket{j,i}} + \frac{1}{\sqrt{2m}}\ket{0}\]
One can verify that such a unitary transformation exists by checking the inner products:
\begin{enumerate}[1)]
\item for any $i\in[2m]$,
  \[\bra{i}U_2^\dagger U_2\ket{i} = \sum_{j>i}{\frac{1}{2m}} + \sum_{j<i}{\frac{1}{2m}} + \frac{1}{2m} = 1.\]
\item for any $i,j\in[2m]$, $i\neq j$,
  \[\begin{split}\bra{j}U_2^\dagger U_2\ket{i} = \left(\sum_{l>j}{\frac{1}{2m}\bra{j,l}} - \sum_{l<j}{\frac{1}{2m}\bra{l,j}} + \frac{1}{2m}\bra{0}\right)\cdot \\
  \left(\sum_{l>i}{\frac{1}{2m}\ket{i,l}} - \sum_{l<i}{\frac{1}{2m}\ket{l,i}} + \frac{1}{2m}\ket{0}\right) = 0
    \end{split}
  \]
\end{enumerate}
The resulting quantum state is
\[\sum_{i=1}^{2m}{\frac{\hat x_i}{\sqrt{2m}}\ket{i}} \xrightarrow{U_2} \sum_{i=1}^{2m}{\frac{\hat x_i}{2m}\ket{0}} + \sum_{i<j}{\frac{\hat x_i-\hat x_j}{2m}\ket{i,j}}.\]
If we measure the state and get $\ket{0}$, then $EXACT_m^{2m}(x)=0$. If on the other hand we get $\ket{i,j}$, then $x_i \neq x_j$ and $EXACT_m^{2m}(x) = EXACT_{m-1}^{2m-2}(x\setminus \{x_i, x_j\})$, therefore we can use our algorithm for $EXACT_{m-1}^{2m-2}$.

\end{proof}

Note that we can delay the measurements by using $\ket{i,j}$ as a starting state for the recursive call of the algorithm.

For the sake of completeness, we include the following corollary already given in \cite{MJM11}:
\begin{corollary}
\cite{MJM11}
\[Q_E(EXACT_k^n) \leq \max{\{k, n-k\}}\]
\end{corollary}
\begin{proof}
  Assume that $k < \frac{n}{2}$. The other case is symmetric. Then we append the input $x$ with $n-2k$ ones producing $x'$ and call $EXACT_{n-k}^{2n-2k}(x')$. Then concluding that there are $n-k$ ones in $x'$ is equivalent to there being $(n-k)-(n-2k)=k$ ones in the original input $x$.
\end{proof}

The lower bound can be established by the following fact:
\begin{proposition}
\label{subfun}
If $g$ is a partial function such that $g(x)=f(x)$ whenever $g$ is defined on $x$, then $Q_E(g) \leq Q_E(f)$.
\end{proposition}

\begin{proposition}
\[Q_E(EXACT_k^n) \geq \max{\{k, n-k\}}\]
\end{proposition}
\begin{proof}
  Assume that $k \leq \frac{n}{2}$. The other case is symmetric. Define 
\[g(x_{k+1}, \ldots, x_n) = EXACT_k^n(1, \ldots, 1, x_{k+1}, \ldots, x_n).\]
  Observe that $g$ is in fact negation of the $OR$ function on $n-k$ bits which we know \cite{BBC+98} to take $n-k$ queries to compute. Therefore by virtue of Proposition \ref{subfun} no algorithm for $EXACT_k^n$ may use less than $n-k$ queries.
\end{proof}

\section{Algorithm for THRESHOLD}

We will abbreviate THRESHOLD as $Th$.

\begin{definition}
The function $Th_k^n$ is a Boolean function of $n$ variables being true iff \emph{at least} $k$ of the variables are equal to $1$.
\end{definition}

The function $Th_{k+1}^{2k+1}$ is commonly referred to as $MAJ_{2k+1}$ or $MAJORITY_{2k+1}$ because it is equal to the majority of values of input variables.

Remarkably an approach similar to the one used for $EXACT$ works in this case as well.


\begin{theorem}
\[Q_E(MAJ_{2k+1})\leq k+1.\]
\end{theorem}
\begin{proof}
Again, a recursive solution is constructed as follows. The base case $k=0$ is trivial to perform with one query, because the function returns the value of the single variable. The recursive step $k=m$ shares the states, unitary transformation $U_1$ and the query with our algorithm for $EXACT$, but the unitary $U_2$ is slightly different:
\[U_1\ket{0} \rightarrow \sum_{i=1}^{2m+1}{\frac{1}{\sqrt{2m+1}}\ket{i}}.\]
\[\sum_{i=1}^{2m+1}{\frac{1}{\sqrt{2m+1}}\ket{i}} \xrightarrow{Q} \sum_{i=1}^{2m+1}{\frac{\hat x_i}{\sqrt{2m+1}}\ket{i}}.\]
\[U_2\ket{i} = \sum_{j>i}{\frac{\sqrt{2m-1}}{2m}\ket{i,j}} - \sum_{j<i}{\frac{\sqrt{2m-1}}{2m}\ket{j,i}} + \sum_{j\neq i}{\frac{1}{2m}\ket{j}}.\]
The resulting state is
\[\sum_{i=1}^{2m+1}{\frac{\hat x_i}{\sqrt{2m+1}}\ket{i}} \xrightarrow{U_2} \sum_{i=1}^{2m+1}{\sum_{j\neq i}{\frac{\hat x_j}{2m\sqrt{2m+1}}}\ket{i}} + \sum_{i<j}{\frac{(\hat x_i-\hat x_j)\sqrt{2m-1}}{2m\sqrt{2m+1}}\ket{i,j}}.\]
We perform a complete measurement. There are two kinds of outcomes:
\begin{enumerate}[1)]
\item If we get state $\ket{i}$, then either
  \begin{enumerate}[a)]
  \item $x_i$ is the value in the majority which according to the polynomial $\sum_{j\neq i}{\hat x_j}$ not being zero implies that in $x\setminus \{x_i\}$ the number of ones is greater than the number of zeroes by at least 2; or
  \item $x_i$ is a value in the minority.
  \end{enumerate}
In both of these cases, for all $j:j\neq i$ it is true that $MAJ_{2m+1}(x)=MAJ_{2m-1}(x \setminus \{x_i, x_j\})$. Therefore, we can solve both cases by removing $x_i$ and one other arbitrary input value and calculating majority from the remaining values.
\item If we get state $\ket{i, j}$, then it is even better: we know that $x_i \neq x_j$ and therefore $MAJ_{2m+1}(x)=MAJ_{2m-1}(x\setminus \{x_i, x_j\})$.
\end{enumerate}
\end{proof}

\begin{corollary}
If $0 < k < n$, then
\[Q_E(Th_k^n)\leq \max{\{k, n-k+1\}}.\]
\end{corollary}
\begin{proof}
  Assume that $k \leq \frac{n}{2}$. The other case is symmetric. Then we append the input $x$ with $n-2k+1$ ones producing $x'$ and call $MAJ_{2n-2k+1}(x')$. Then $x'$ containing at least $n-k+1$ ones is equivalent to $x$ containing at least $(n-k+1)-(n-2k+1)=k$ ones.
\end{proof}


\begin{proposition}
\[Q_E(Th_k^n) \geq \max{\{k, n-k+1\}}\]
\end{proposition}
\begin{proof}
  Assume that $k \leq \frac{n}{2}$. The other case is symmetric. Define 
\[g(x_k, x_{k+1}, \ldots, x_n) = Th_k^n(1, \ldots, 1, x_k, x_{k+1}, \ldots, x_n).\]
Observe that $g$ is in fact the $OR$ function on $n-k+1$ bits which we know \cite{BBC+98} takes $n-k+1$ queries to compute. Therefore by virtue of Proposition \ref{subfun} no algorithm for $Th_k^n$ may use less than $n-k+1$ queries.
\end{proof}

\section{Conclusion}
Coming up with exact quantum algorithms that are substantially better than any 
classical algorithm has been a difficult open problem. Until a few months ago, no example of
total Boolean function with $Q_E(f)<D(f)/2$ was known and the examples of functions with $Q_E(f)=D(f)/2$
were almost all based on one idea: applying 1-query quantum algorithm for $x_1\oplus x_2$
as a subroutine.

The first exact quantum algorithm with $Q_E(f)<D(f)/2$ (for a total $f$) was constructed in
\cite{Amb12}. However, no symmetric function with $Q_E(f)<D(f)/2$ is known. It has been proven that if $f(x)$ is a symmetric, non-constant function of $n$ variables, then $Q_E(f)\geq n/2-o(n)$ \cite{ZGR97,BdW02}.

In this paper, we construct exact quantum algorithms for two symmetric functions: $EXACT$ and $THRESHOLD$. Both of those algorithms achieve $Q_E(f) = D(f)/2$
(exactly or in the limit) and use new ideas. At the same time, our algorithms are quite simple and easy to understand.

The main open problem is to come with more algorithmic techniques for constructing exact quantum algorithms. Computer experiments via semidefinite optimization \cite{MJM11} show that there are many functions for which exact quantum algorithms are better
than deterministic algorithms. Yet, in many of those case, the only way to construct 
these algorithms is by searching the space of all quantum algorithms, using
semidefinite optimization as the search tool.

For example, from the calculations in \cite{MJM11} (based on semidefinite optimization) it is apparent that there are 3 symmetric functions of 6 variables for which $Q_E(f)=3$: $PARITY$, $EXACT_3^6$ and $EXACT_{2,4}^6$ (exactly 2 or 4 of 6 variables are equal to 1). 

Unlike for the first two functions, we are not aware of any simple quantum algorithm or lower bounds for $EXACT_{2,4}^6$. Based on the evidence from semidefinite optimization, we conjecture that if $n$ is even and $2k<n$ then the quantum query complexity of $EXACT_{k,n-k}^n$ is $n-k-1$. In particular, this would mean that the complexity of
$EXACT_{n/2-1, n/2+1}^n$ is $\frac{n}{2}$ and this function also achieves a gap of 
$Q_E(f)=D(f)/2$.

At the moment, we know that this conjecture is true for $k=0$ and $k=1$. Actually, 
both of those cases can be solved by a classical algorithm which uses the 1-query algorithm 
for $x_1\oplus x_2$ as a quantum subroutine. This approach fails for $k\geq 1$ and it seems
that the approach in the current paper is also not sufficient --- without a substantial new component.

\phantomsection

\end{document}